\documentclass{aims}
\usepackage{amsmath}
  \usepackage{paralist}
  \usepackage{bm}
  \usepackage{graphics} 
\usepackage{graphicx}  
 \usepackage[colorlinks=true]{hyperref}
\hypersetup{urlcolor=blue, citecolor=red}

\def\currentvolume{X}
 \def\currentissue{X}
  \def\currentyear{200X}
   \def\currentmonth{XX}
    \def\ppages{X--XX}
     \def\DOI{10.3934/xx.xx.xx.xx}

\newtheorem{theorem}{Theorem}[section]

\newtheorem{lemma}[theorem]{Lemma}

\theoremstyle{definition}

\newtheorem{remark}{Remark}

\title[On the thermal relaxation of a dense gas in a closed system] 
      {On the thermal relaxation of a dense gas described by the modified Enskog
equation in a closed system in contact with a heat bath}

\author[Shigeru Takata]{}


 \email{takata.shigeru.4a@kyoto-u.ac.jp}

\thanks{$^*$ Corresponding author: Shigeru Takata}

\begin{document}
\maketitle

\centerline{\scshape Shigeru TAKATA$^*$}
\medskip
{\footnotesize
 \centerline{Department of Aeronautics and Astronautics, Kyoto University,}
   \centerline{Kyoto daigaku-katsura, Nishikyo-ku, Kyoto 615-8540, Japan}
} 

%

\bigskip

 \centerline{(Communicated by the associate editor name)}

\begin{abstract}
The thermal relaxation of a dense gas described by the modified Enskog
equation is studied for a closed system in contact with a heat bath.
As in the case of the Boltzmann equation, the Helmholtz free energy
$\mathcal{F}$ that decreases monotonically in time is found under
the conventional kinetic boundary condition that satisfies the Darrozes--Guiraud
inequality. 
{The extension to the modified Enskog--Vlasov equation is also presented.}
\end{abstract}

\section{Introduction\label{sec:Introduction}}

Behavior of ideal gases is well described by the Boltzmann equation
for the entire range of the Knudsen number, the ratio of the mean
free path of gas molecules to a characteristic length of the system.
The kinetic theory based on the Boltzmann equation and its model equations
has been applied successfully to analyses of various gas flows in
low pressure circumstances, micro-scale gas flows, and gas flows caused
by the evaporation/condensation at the gas-liquid interface. 

The extension of the kinetic theory to non-ideal gases would go back to the dates
of Enskog \cite{E72}. He took account of the displacement effect of molecules
in collision integrals for a hard-sphere gas and proposed
a kinetic equation that is nowadays called the (original) Enskog equation.
In the original Enskog equation, there appears a weight function that
represents an equilibrium correlation function at the contact point of two colliding
molecules. On one hand, satisfactory outcomes {of} the original Enskog equation, 
such as the dense gas effects on the transport properties,
{led} to recent developments of numerical algorithms \cite{F97,WLRZ16} and their applications to physical problems, e.g., \cite{F99,KKW14,WLRZ16,FGLS18}.
On the other hand, {the intuitive} choice of the correlation 
was recognized to cause some difficulties in recovering the H theorem, as well as
the Onsager reciprocity in the case of mixtures, 
and triggered off further intensive studies on the foundation of the equation around from late 60's to early 80's, 
see, e.g., \cite{BE73,R78,MGB18} and references therein. 

Among many efforts in the above-mentioned period, 
Resibois~\cite{R78} succeeded to prove the H theorem, not
for the original but for the modified Enskog equation \cite{BE73} equipped with another form of correlation function. 
In most cases, including the work of Resibois, the H theorem {was} discussed mainly for periodic or unbounded spatial domains, or for cases where { the influence of a boundary was not necessary to consider}, e.g.,~\cite{R78,GG80,BB18,KS82}.
Rather recently, a proof  {was given by Maynar~et~al.~\cite{MGB18} for an isolated system, assuming the specular reflection condition, where special care was directed to a restriction on the range of collision integral} near the boundary. It seems, however, that the thermal relaxation in contact with a heat bath receives little attention in the literature, {despite the fact} that it is one of the fundamental issues in {the} thermo-statistical physics. 
Although the interaction with the thermostat boundary is considered in a recent monograph of Dorfman~et~al.~\cite{DBK21}, we are not aware of a direct discussion on the thermal relaxation of a dense gas in a closed system in contact with a heat bath in the context of the modified Enskog equation.

In the present paper, we would like to fill the gap by a simple argument and to show that, if the boundary condition satisfies the Darrozes--Guiraud
inequality~\cite{DG66} that is conventionally required in {the} kinetic theory, the Helmholtz free energy $\mathcal{F}$
that decreases monotonically in time can be found for a closed system
described by the modified Enskog equation as in the case of the Boltzmann
equation.

\section{Problem and formulation\label{sec:Modified-Enskog-equation}}

Consider a dense gas in a domain that is surrounded by a simple resting
solid wall kept at a uniform temperature $T_{w}$, i.e., a heat bath
with temperature $T_{w}$. We will study the relaxation of the gas toward
a thermal equilibrium state with the heat bath under the following
assumptions:
\begin{enumerate}
\item The behavior of the gas is described by the modified Enskog equation
for a single species gas;
\item The gas molecules are hard spheres with a common diameter $\sigma$
and mass $m$ and the collisions among themselves are elastic;
\item The velocity distribution of gas molecules reflected on the surface
of the heat bath is described by the kinetic boundary condition that
is conventionally used for the Boltzmann equation, the details of
which will be given in \eqref{KBC}.
\end{enumerate}
Let $D$ be a fixed spatial domain that the centers of molecules of
a gas can occupy. Let {$t$, $\bm{X}$ and $\bm{Y}$, and $\bm{\xi}$}
be a time, spatial {positions}, and a molecular velocity, respectively.
Then, denoting the one-particle distribution function of gas molecules
by $f(t,\bm{X},\bm{\xi}$) {and the correlation function  
by $g(t,\bm{X},\bm{Y})$}, the modified Enskog equation is written
as \begin{subequations}\label{MEE}
\begin{align}
 & \frac{\partial f}{\partial t}+\xi_{i}\frac{\partial f}{\partial X_{i}}=J_{ME}(f)\equiv J_{ME}^{G}(f)-J_{ME}^{L}(f),\quad \mathrm{for\ }\bm{X}\in D,\displaybreak[0]\label{eq:2.1}\\
 & J_{ME}^{G}(f)\equiv\frac{\sigma^{2}}{m}\int {g(\bm{X}_{\sigma\bm{\alpha}}^{+},\bm{X})f_{*}^{\prime}(\bm{X}_{\sigma\bm{\alpha}}^{+})f^{\prime}(\bm{X})}V_{\alpha}\theta(V_{\alpha})d\Omega(\bm{\alpha})d\bm{\xi}_{*},\displaybreak[0]\label{eq:2.2}\\
 & J_{ME}^{L}(f)\equiv\frac{\sigma^{2}}{m}\int {g(\bm{X}_{\sigma\bm{\alpha}}^{-},\bm{X})f_{*}(\bm{X}_{\sigma\bm{\alpha}}^{-})f(\bm{X})}V_{\alpha}\theta(V_{\alpha})d\Omega(\bm{\alpha})d\bm{\xi}_{*},\label{eq:2.3}
\end{align}
\end{subequations}where $\bm{X}_{\bm{x}}^{\pm}=\bm{X}\pm\bm{x}$,
$\bm{\alpha}$ is a unit vector, 
{\begin{equation}
\theta(x)=\begin{cases}
1, & x\ge0\\
0, & x<0
\end{cases},
\end{equation}%
$d\Omega(\bm{\alpha})$ is a
solid angle element in the direction of $\bm{\alpha}$,
and the following notation convention is used:}
\begin{align}
 & \begin{cases}
{
f(\bm{X})=f(\bm{X},\bm{\xi}),\ f^{\prime}(\bm{X})=f(\bm{X},\bm{\xi}^{\prime})},\\
{
f_{*}(\bm{X}_{\sigma\bm{\alpha}}^{-})=f(\bm{X}_{\sigma\bm{\alpha}}^{-},\bm{\xi}_{*}),\ f_{*}^{\prime}(\bm{X}_{\sigma\bm{\alpha}}^{-})=f(\bm{X}_{\sigma\bm{\alpha}}^{-},\bm{\xi}_{*}^{\prime})},
\end{cases}\displaybreak[0]\\
 & \bm{\xi}^{\prime}=\bm{\xi}+V_{\alpha}\bm{\alpha},\quad\bm{\xi}_{*}^{\prime}=\bm{\xi}_{*}-V_{\alpha}\bm{\alpha},\quad V_{\alpha}=\bm{V}\cdot\bm{\alpha},\quad\bm{V}=\bm{\xi_{*}}-\bm{\xi}.\label{eq:2.5}
\end{align}
{ Here and in what follows, the argument $t$ is suppressed, unless confusion is anticipated.}
{Our correlation function $g$ is adjusted to the domain $D$ in such a way that the usual correlation function $g_{2}(t,\bm{X},\bm{Y})$ is modified as
\begin{subequations}\label{g-g2}\begin{align}
 & g(t,\bm{X},\bm{Y})=g_{2}(t,\bm{X},\bm{Y})\chi_{D}(\bm{X})\chi_{D}(\bm{Y}),\displaybreak[0]\label{eq:g_g2}\\
 & \chi_{D}(\bm{X})=\begin{cases}
1, & \bm{X}\in D\\
0, & \mbox{otherwise}
\end{cases},
\end{align}\end{subequations}
where $\chi_D$ plays the same role as the Heaviside function $\theta$.
Consequently, the range of integration in
(\ref{eq:2.2}) and (\ref{eq:2.3}) can be treated as the whole space
of $\bm{\xi}_*$ and all directions of $\bm{\alpha}$ even near the
surface of the domain $\partial D$.
In contrast to the original Enskog equation, 
$g_2$ takes a complicated form that requires further supplemental notation.
For the moment, it suffices to mention that
$g_{2}$ has a symmetric property $g_{2}(t,\bm{X},\bm{Y})=g_{2}(t,\bm{Y},\bm{X})$
and is a functional of a gas density 
\begin{equation}
\rho=\int fd\bm{\xi}.
\end{equation}
Therefore \eqref{MEE} is closed as the equation for $f$.
By \eqref{g-g2}, $g$ has}
the same symmetric property as $g_{2}$:
\begin{equation}
g(t,\bm{X},\bm{Y})=g(t,\bm{Y},\bm{X}).
\end{equation}
{Further details of $g_2$ can be found in Appendix~\ref{sec:Correlation-function}.}

The boundary condition is applied on the surface $\partial D$ of
the domain $D$:\begin{subequations}\label{KBC}
\begin{equation}
f(t,\bm{X},\bm{\xi})=\int_{\bm{\xi}_{*}\cdot\bm{n}<0}K(\bm{\xi},\bm{\xi}_{*}|\bm{X})f(t,\bm{X},\bm{\xi}_{*})d\bm{\xi}_{*},\quad(\bm{\xi}\cdot\bm{n}>0,\ \bm{X}\in\partial D),\label{eq:2.10}
\end{equation}
where $K(\bm{\xi},\bm{\xi}_{*}|\bm{X})$ is a scattering kernel assumed
to be time-independent, $\bm{n}$ is the inward unit normal to the
surface $\partial D$ at position $\bm{X}$, and the boundary is assumed
to be at rest. The following properties are conventionally supposed for a kinetic
boundary condition: \cite{S07}
\begin{enumerate}
\item Non-negativeness:
\begin{equation}
K(\bm{\xi},\bm{\xi}_{*}|\bm{X})\ge0,\quad(\bm{\xi}\cdot\bm{n}>0,\ \bm{\xi}_{*}\cdot\bm{n}<0);
\end{equation}
\item Normalization:
\begin{equation}
\int_{\bm{\xi}\cdot\bm{n}>0}\Big|\frac{\bm{\xi}\cdot\bm{n}}{\bm{\xi}_{*}\cdot\bm{n}}\Big|K(\bm{\xi},\bm{\xi}_{*}|\bm{X})d\bm{\xi}=1,\quad(\bm{\xi}_{*}\cdot\bm{n}<0),\label{eq:3.12}
\end{equation}
where the integrand in (\ref{eq:3.12}) is the so-called reflection
probability density. Equation~(\ref{eq:3.12}) implies that the boundary
$\partial D$ is impermeable;
\item Preservation of equilibrium: The resting Maxwellian $f_{w}$ characterized
by the surface temperature $T_{w}$, i.e.,
\begin{equation}
f_{w}=\frac{a}{(2\pi RT_{w})^{3/2}}\exp(-\frac{\bm{\xi}^{2}}{2RT_{w}}),\label{eq:fw}
\end{equation}
with $a(>0)$ being arbitrary, satisfies the boundary condition (\ref{eq:2.10}),
and the other Maxwellians do not satisfy (\ref{eq:2.10}).
\end{enumerate}
\end{subequations}The diffuse reflection, the Maxwell, and the Cercignani--Lampis
condition \cite{CL71,C88,S07} that are widely used for the Boltzmann
equation are specific examples of \eqref{KBC}. Note that the uniqueness in the third
property listed above excludes the adiabatic boundary such as the
specular reflection condition. As to the H theorem for the specular
reflection case, the reader is referred to \cite{MGB18}.

The form of the modified Enskog equation \eqref{MEE} is identical
to the one for a confined isolated system discussed in \cite{MGB18}. 
In our formulation, { $\chi_{D}$} is used
to make simpler the integration range near
the surface $\partial D$.

\section{Collisional contributions to the momentum and the energy transport\label{sec:Collisional-contributions}}

Before going into details, we recall three types of operation that
are useful in the transformation of the moments of collision integrals:
\begin{description}
\item [{(I)}] to exchange the letters $\bm{\xi}$ and $\bm{\xi}_{*}$; 
\item [{(II)}] to reverse the direction of $\bm{\alpha}$ (or introduce $\bm{\beta}=-\bm{\alpha}$);
\item [{(III)}] to change the integration variables from $(\bm{\xi},\bm{\xi}_{*},\bm{\alpha})$
to $(\bm{\xi}^{\prime},\bm{\xi}_{*}^{\prime},\bm{\alpha})$ and then to
change the letters $(\bm{\xi}^{\prime},\bm{\xi}_{*}^{\prime})$ to
$(\bm{\xi},\bm{\xi}_{*})$.
\end{description}
These operations will be used also in Sec.~\ref{sec:H-function}.

We then notice that, by (III) and (II),
\begin{equation}
\int{\varphi} J_{ME}^{G}(f)d\bm{\xi}=\int{\varphi}^{\prime}J_{ME}^{L}(f)d\bm{\xi},\label{eq:4.1}
\end{equation}
holds for any ${\varphi}(\bm{\xi})$ and thus
\begin{align}
 & \frac{m}{\sigma^{2}}\int{\varphi}(\bm{\xi})J_{ME}(f)d\bm{\xi}\nonumber \\
= & \int({\varphi}^{\prime}-{\varphi})g(\bm{X}_{\sigma\bm{\alpha}}^{-},\bm{X})f_{*}(\bm{X}_{\sigma\bm{\alpha}}^{-})f(\bm{X})V_{\alpha}\theta(V_{\alpha})d\Omega(\bm{\alpha})d\bm{\xi}_{*}d\bm{\xi}.\label{eq:mom_J}
\end{align}

First, it is obvious from (\ref{eq:mom_J}) with ${\varphi}=1$ that $\int J_{ME}(f)d\bm{\xi}=0$.
Hence, the continuity equation is obtained by the integration of (\ref{eq:2.1})
with respect to $\bm{\xi}$:
\begin{equation}
\frac{\partial\rho}{\partial t}+\frac{\partial}{\partial\bm{X}}\cdot(\rho\bm{v})=0.\label{eq:continuity}
\end{equation}
Here $\bm{v}$ (or $v_i$) is a flow velocity defined by
\begin{equation}
v_{i}=\frac{1}{\rho}\int\xi_{i}fd\bm{\xi}.\label{eq:velocity}
\end{equation}
 
Next, consider two kinds of collision invariants
$\psi$ as ${\varphi}$ in (\ref{eq:mom_J}): (i) $\psi(\bm{\xi})=\xi_{i}$ and (ii) $\psi(\bm{\xi})=\bm{\xi}^{2}/2$.
One of the main qualitative differences from the Boltzmann equation
is that $\psi$-moment of the collision term does not vanish in general.
For both (i) and (ii), (\ref{eq:mom_J}) with ${\varphi}=\psi$ can be
transformed as
\begin{subequations}\label{eq:4.3}\begin{align}
    \frac{m}{\sigma^{2}}\int\psi(\bm{\xi})J_{ME}(f)d\bm{\xi}
= & \int(\psi_{*}^{\prime}-\psi_{*})g(\bm{X}_{\sigma\bm{\beta}}^{+},\bm{X})f(\bm{X}_{\sigma\bm{\beta}}^{+})f_{*}(\bm{X})V_{\beta}\theta(V_{\beta})d\Omega(\bm{\beta})d\bm{\xi}d\bm{\xi}_{*}\displaybreak[0]\nonumber \\
= & \int(\psi-\psi^{\prime})g(\bm{X}_{\sigma\bm{\beta}}^{+},\bm{X})f(\bm{X}_{\sigma\bm{\beta}}^{+})f_{*}(\bm{X})V_{\beta}\theta(V_{\beta})d\Omega(\bm{\beta})d\bm{\xi}d\bm{\xi}_{*},
\label{eq:4.3a}
\end{align}
where (I) and (II) are used at the first equality, 
while $\psi$+$\psi_{*}$=$\psi^{\prime}$+$\psi_{*}^{\prime}$
is used at the second equality.
Combining \eqref{eq:4.3a} and (\ref{eq:mom_J})
for ${\varphi}=\psi$ gives
\begin{align}
 \frac{m}{\sigma^{2}}\int\psi(\bm{\xi})J_{ME}(f)d\bm{\xi}
= & \frac{1}{2}\int(\psi^{\prime}-\psi)\{g(\bm{X}_{\sigma\bm{\alpha}}^{-},\bm{X})f_{*}(\bm{X}_{\sigma\bm{\alpha}}^{-})f(\bm{X})\nonumber \\
 & -g(\bm{X}_{\sigma\bm{\alpha}}^{+},\bm{X})f(\bm{X}_{\sigma\bm{\alpha}}^{+})f_{*}(\bm{X})\}V_{\alpha}\theta(V_{\alpha})d\Omega(\bm{\alpha})d\bm{\xi}_{*}d\bm{\xi}.\label{eq:4.3b}
\end{align}\end{subequations}
Since
\begin{align}
 & g(\bm{X}_{\sigma\bm{\alpha}}^{-},\bm{X})f_{*}(\bm{X}_{\sigma\bm{\alpha}}^{-})f(\bm{X})-g(\bm{X}_{\sigma\bm{\alpha}}^{+},\bm{X})f(\bm{X}_{\sigma\bm{\alpha}}^{+})f_{*}(\bm{X})\displaybreak[0]\notag\\
= & -\int_{0}^{\sigma}\frac{\partial}{\partial\lambda}g(\bm{X}_{\lambda\bm{\alpha}}^{+},\bm{X}_{(\lambda-\sigma)\bm{\alpha}}^{+})f_{*}(\bm{X}_{(\lambda-\sigma)\bm{\alpha}}^{+})f(\bm{X}_{\lambda\bm{\alpha}}^{+})d\lambda\displaybreak[0]\notag\\
= & -\bm{\nabla}\cdot\int_{0}^{\sigma}\bm{\alpha}g(\bm{X}_{\lambda\bm{\alpha}}^{+},\bm{X}_{(\lambda-\sigma)\bm{\alpha}}^{+})f_{*}(\bm{X}_{(\lambda-\sigma)\bm{\alpha}}^{+})f(\bm{X}_{\lambda\bm{\alpha}}^{+})d\lambda,\label{eq:15}
\end{align}
\eqref{eq:4.3b} gives rise to the notion of collisional contributions to the stress tensor
$p_{ij}^{(c)}$ and the heat flow $q_{i}^{(c)}$ defined as
\begin{subequations}\label{gradient}
\begin{align}
p_{ij}^{(c)}= & \frac{\sigma^{2}}{2m}\int {\int_{0}^{\sigma}}\alpha_{i}\alpha_{j}V_{\alpha}^{2}\theta(V_{\alpha})\nonumber \\
 & \quad g(\bm{X}_{\lambda\bm{\alpha}}^{+},\bm{X}_{(\lambda-\sigma)\bm{\alpha}}^{+})f_{*}(\bm{X}_{(\lambda-\sigma)\bm{\alpha}}^{+})f(\bm{X}_{\lambda\bm{\alpha}}^{+}){d\lambda} d\Omega(\bm{\alpha})d\bm{\xi}_{*}d\bm{\xi},\displaybreak[0]\label{eq:4.5}\\
q_{i}^{(c)}= & -p_{ij}^{(c)}v_{j}+\frac{\sigma^{2}}{4m}\int{\int_{0}^{\sigma}}\alpha_{i}[(\bm{\xi}+\bm{\xi}_{*})\cdot\bm{\alpha}]V_{\alpha}^{2}\theta(V_{\alpha})\nonumber \\
 & \quad g(\bm{X}_{\lambda\bm{\alpha}}^{+},\bm{X}_{(\lambda-\sigma)\bm{\alpha}}^{+})f_{*}(\bm{X}_{(\lambda-\sigma)\bm{\alpha}}^{+})f(\bm{X}_{\lambda\bm{\alpha}}^{+}){d\lambda} d\Omega(\bm{\alpha})d\bm{\xi}_{*}d\bm{\xi}\displaybreak[0]\nonumber \\
= & \frac{\sigma^{2}}{4m}\int{\int_{0}^{\sigma}}\alpha_{i}[(\bm{c}+\bm{c}_{*})\cdot\bm{\alpha}]V_{\alpha}^{2}\theta(V_{\alpha})\nonumber \\
 & \quad g(\bm{X}_{\lambda\bm{\alpha}}^{+},\bm{X}_{(\lambda-\sigma)\bm{\alpha}}^{+})f_{*}(\bm{X}_{(\lambda-\sigma)\bm{\alpha}}^{+})f(\bm{X}_{\lambda\bm{\alpha}}^{+}){d\lambda} d\Omega(\bm{\alpha})d\bm{\xi}_{*}d\bm{\xi},\label{eq:4.6}
\end{align}\end{subequations}
see e.g., \cite{CL88,F99}. Here $\bm{c}=\bm{\xi}-\bm{v}$, $\bm{c}_{*}=\bm{\xi}_{*}-\bm{v}$,
and 
\begin{equation}
\psi^{\prime}-\psi=\begin{cases}
V_{\alpha}\alpha_{i}, & {\displaystyle (\psi=\xi_{i}),}\\
{\displaystyle \frac{1}{2}V_{\alpha}(\bm{\xi}+\bm{\xi}_{*})\cdot\bm{\alpha},} & ({\displaystyle \psi=\frac{1}{2}\bm{\xi}^{2}),}
\end{cases}\label{psi_p}
\end{equation}
have been used. 
{Note that, thanks to the factor $\chi_D$ in $g$,
the range of integration with respect to $\lambda$ is simply from $0$ to $\sigma$, 
regardless of the position $\bm{X}$ in $D$. }

To summarize, two expressions for the same quantity have been obtained.
For the quantity related to the energy, 
\begin{subequations}\label{ET}\begin{align}
 & \int\frac{1}{2}\bm{\xi}^{2}J_{ME}(f)d\bm{\xi}\nonumber \\
= & -\frac{\sigma^{2}}{2m}\int[(\bm{\xi}+\bm{\xi}_{*})\cdot\bm{\alpha}]V_{\alpha}^{2}\theta(V_{\alpha})g(\bm{X}_{\sigma\bm{\alpha}}^{+},\bm{X})f(\bm{X}_{\sigma\bm{\alpha}}^{+})f_{*}(\bm{X})d\Omega(\bm{\alpha})d\bm{\xi}d\bm{\xi}_{*},\label{eq:4.7}
\end{align}
and
\begin{equation}
  \int\frac{1}{2}\bm{\xi}^{2}J_{ME}(f)d\bm{\xi}
=-\frac{\partial}{\partial X_{i}}(p_{ij}^{(c)}v_{j}+q_{i}^{(c)}),\label{eq:4.8}
\end{equation}\end{subequations}
see \eqref{eq:4.3a} with \eqref{psi_p} and \eqref{gradient}; for the quantity related to the momentum,
\begin{subequations}\label{MT}\begin{align}
 & \int\xi_{i}J_{ME}(f)d\bm{\xi}\nonumber \\
= & -\frac{\sigma^{2}}{m}\int\alpha_{i}V_{\alpha}^{2}\theta(V_{\alpha})g(\bm{X}_{\sigma\bm{\alpha}}^{+},\bm{X})f(\bm{X}_{\sigma\bm{\alpha}}^{+})f_{*}(\bm{X})d\Omega(\bm{\alpha})d\bm{\xi}d\bm{\xi}_{*},\label{eq:4.9}
\end{align}
and
\begin{equation}
\int\xi_{i}J_{ME}(f)d\bm{\xi}=-\frac{\partial}{\partial X_{j}}p_{ij}^{(c)},\label{eq:4.10}
\end{equation}\end{subequations}
see \eqref{eq:4.3a} with \eqref{psi_p} and \eqref{eq:4.5}.

Finally by integrating (\ref{eq:4.7}) over the domain $D$ and recalling
(\ref{eq:g_g2}), it is seen that
\begin{align}
 & \int_{D}\int\frac{1}{2}\bm{\xi}^{2}J_{ME}(f)d\bm{\xi}d\bm{X}\nonumber \\
= & -\frac{\sigma^{2}}{2m}\int[(\bm{\xi}+\bm{\xi}_{*})\cdot\bm{\alpha}]V_{\alpha}^{2}\theta(V_{\alpha})g(\bm{X},\bm{X}_{\sigma\bm{\alpha}}^{-})f(\bm{X})f_{*}(\bm{X}_{\sigma\bm{\alpha}}^{-})d\bm{\xi}d\bm{\xi}_{*}d\Omega(\bm{\alpha})d\bm{X}\displaybreak[0]\nonumber \\
= & \frac{\sigma^{2}}{2m}\int[(\bm{\xi}+\bm{\xi}_{*})\cdot\bm{\beta}]V_{\beta}^{2}\theta(-V_{\beta})g(\bm{X},\bm{X}_{\sigma\bm{\beta}}^{+})f(\bm{X})f_{*}(\bm{X}_{\sigma\bm{\beta}}^{+})d\bm{\xi}d\bm{\xi}_{*}d\Omega(\bm{\beta})d\bm{X}\displaybreak[0]\nonumber \\
= & \frac{\sigma^{2}}{2m}\int[(\bm{\xi}+\bm{\xi}_{*})\cdot\bm{\beta}]V_{\beta}^{2}\theta(V_{\beta})g(\bm{X},\bm{X}_{\sigma\bm{\beta}}^{+})f_{*}(\bm{X})f(\bm{X}_{\sigma\bm{\beta}}^{+})d\bm{\xi}_{*}d\bm{\xi}d\Omega(\bm{\beta})d\bm{X}\displaybreak[0]\nonumber \\
= & -\int_{D}\int\frac{1}{2}\bm{\xi}^{2}J_{ME}(f)d\bm{\xi}d\bm{X},\label{eq:4.11}
\end{align}
where the position is shifted by $-\sigma\bm{\alpha}$ at the first
equality, (II) and (I) are applied respectively at the second and the third equality,
and (\ref{eq:4.7}) is used at the last equality. Hence 
\begin{equation}
\int_{D}\int\frac{1}{2}\bm{\xi}^{2}J_{ME}(f)d\bm{\xi}d\bm{X}=0,
\end{equation}
and by (\ref{eq:4.8}) 
\begin{equation}
-\int_{D}\frac{\partial}{\partial X_{i}}(p_{ij}^{(c)}v_{j}+q_{i}^{(c)})d\bm{X}=\int_{\partial D}(p_{ij}^{(c)}v_{j}+q_{i}^{(c)})n_{i}dS=0.
\end{equation}
Here the divergence theorem has been used and $\bm{n}$
is the inward unit normal to the surface $\partial D$. In the same
way, it can be shown that
\begin{equation}
\int_{D}\int\xi_{i}J_{ME}(f)d\bm{\xi}d\bm{X}=0,
\end{equation}
and by (\ref{eq:4.10})
\begin{equation}
-\int_{D}\frac{\partial}{\partial X_{j}}p_{ij}^{(c)}d\bm{X}=\int_{\partial D}p_{ij}^{(c)}n_{j}dS=0.
\end{equation}
\begin{lemma}
\label{lem:E transfer}In total, there are no collisional contributions
to the momentum and energy transport:
\begin{equation}
\int_{D}\int\xi_{i}J_{ME}(f)d\bm{\xi}d\bm{X}=0,\quad\int_{D}\int\frac{1}{2}\bm{\xi}^{2}J_{ME}(f)d\bm{\xi}d\bm{X}=0.\label{eq:lem1-a}
\end{equation}
Accordingly, there are no collisional contributions to the net momentum
and energy transport to the surface $\partial D$:
\begin{equation}
\int_{\partial D}p_{ij}^{(c)}n_{j}dS=0,\quad\int_{\partial D}(p_{ij}^{(c)}v_{j}+q_{i}^{(c)})n_{i}dS=0.\label{eq:lem1-b}
\end{equation}
In particular, if $D$ is convex, $p_{ij}^{(c)}\equiv0$ and $q_{i}^{(c)}\equiv0$
on the surface $\partial D$.
\end{lemma}

\begin{proof}
Equations (\ref{eq:lem1-a}) and (\ref{eq:lem1-b}) are simply a summary
of the present section. When $D$ is convex, $\chi_{D}(\bm{X}_{+(\lambda-\sigma)\bm{\alpha}})\chi_{D}(\bm{X}_{+\lambda\bm{\alpha}})=0$
for $\bm{X}\in\partial D$, except for the special case that
$\partial D$ is flat at $\bm{X}$. However, the exception occurs only for $\bm{\alpha}$ in the directions tangential to $\partial D$ and thus has no contribution to the integration of the angle in (\ref{eq:4.5}) and (\ref{eq:4.6}). 
\end{proof}
\begin{remark}
Equation~(\ref{eq:lem1-a}) in Lemma~\ref{lem:E transfer}
is physically a natural consequence, since the collisional transport
of momentum and energy comes from interactions within gas molecules.
The collisional stress tensor $p_{ij}^{(c)}$ and heat flow $q_{i}^{(c)}$
are, however, not likely to vanish pointwisely on the surface $\partial D$
if the domain $D$ is not convex.
\end{remark}

\section{H function\label{sec:H-function}}

In this section, we shall recall the discussions on the H theorem
in the literature \cite{R78,MGB18,DBK21}. Consider first the so-called
kinetic part of the H function%
\footnote{To be precise, it is necessary to make the argument of the logarithmic function dimensionless, like $\ln (f/c_0)$ with a constant $c_0$ having the same dimension as $f$.
We, however, leave the argument dimensional 
to avoid additional calculations that do not affect the results.} 
\begin{equation}
\mathcal{H}^{(k)}\equiv\int_{D}\int f{\ln f} d\bm{\xi}d\bm{X},\label{H_kinetic}
\end{equation}
Then, multiplying $1+{\ln f}$ with the modified Enskog
equation (\ref{eq:2.1}) gives
\begin{equation}
\frac{\partial}{\partial t}\langle f{\ln f}\rangle+\frac{\partial}{\partial X_{i}}\langle\xi_{i}f{\ln f}\rangle=\langle J_{ME}(f){\ln f}\rangle,\label{eq:27}
\end{equation}
after the integration with respect to $\bm{\xi}$, where $\langle\bullet\rangle=\int\bullet\ d\bm{\xi}$. The first step
toward the H theorem is to apply (\ref{eq:mom_J}) with ${\varphi=\ln f}$
to the right-hand side:
\begin{align}
 & \frac{m}{\sigma^{2}}\langle J_{ME}(f){\ln f}\rangle\nonumber \\
= & \int\ln[f^{\prime}(\bm{X})/f(\bm{X})]g(\bm{X}_{\sigma\bm{\alpha}}^{-},\bm{X})f_{*}(\bm{X}_{\sigma\bm{\alpha}}^{-})f(\bm{X})V_{\alpha}\theta(V_{\alpha})d\Omega(\bm{\alpha})d\bm{\xi}_{*}d\bm{\xi}.\label{eq:5.2}
\end{align}
Then, the integration of (\ref{eq:5.2}) over the domain $D$ is again
a relevant step for the position shift by $+\sigma\bm{\alpha}$ and
gives
\begin{align}
 & \frac{m}{\sigma^{2}}\int_{D}\langle J_{ME}(f){\ln f}\rangle d\bm{X}\nonumber \\
= & \int\ln[f^{\prime}(\bm{X}_{\sigma\bm{\alpha}}^{+})/f(\bm{X}_{\sigma\bm{\alpha}}^{+})]g(\bm{X},\bm{X}_{\sigma\bm{\alpha}}^{+})f_{*}(\bm{X})f(\bm{X}_{\sigma\bm{\alpha}}^{+})V_{\alpha}\theta(V_{\alpha})d\Omega(\bm{\alpha})d\bm{\xi}_{*}d\bm{\xi}d\bm{X}\displaybreak[0]\nonumber \\
= & \int\ln[f_{*}^{\prime}(\bm{X}_{\sigma\bm{\alpha}}^{-})/f_{*}(\bm{X}_{\sigma\bm{\alpha}}^{-})]g(\bm{X},\bm{X}_{\sigma\bm{\alpha}}^{-})f(\bm{X})f_{*}(\bm{X}_{\sigma\bm{\alpha}}^{-})V_{\alpha}\theta(V_{\alpha})d\Omega(\bm{\alpha})d\bm{\xi}d\bm{\xi}_{*}d\bm{X}\displaybreak[0]\nonumber \\
= & \frac{1}{2}\int\ln\Big(\frac{f_{*}^{\prime}(\bm{X}_{\sigma\bm{\alpha}}^{-})f^{\prime}(\bm{X})}{f_{*}(\bm{X}_{\sigma\bm{\alpha}}^{-})f(\bm{X})}\Big)g(\bm{X},\bm{X}_{\sigma\bm{\alpha}}^{-})f(\bm{X})f_{*}(\bm{X}_{\sigma\bm{\alpha}}^{-})V_{\alpha}\theta(V_{\alpha})d\Omega(\bm{\alpha})d\bm{\xi}d\bm{\xi}_{*}d\bm{X},
\end{align}
where (II) and (I) are applied at the second equality, while the third
line and (\ref{eq:5.2}) are combined at the last equality. Since
for any $x,y>0$
\begin{equation}
x\ln(y/x)\le y-x,\label{eq:log}
\end{equation}
where equality holds if and only if $y=x$,
\begin{equation}
\int_{D}\langle J_{ME}(f){\ln f}\rangle d\bm{X}\le I(t),
\end{equation}
holds, where 
\begin{equation}
I(t)=\frac{\sigma^{2}}{2m}\int g(\bm{X},\bm{X}_{\sigma\bm{\alpha}}^{-})[f_{*}^{\prime}(\bm{X}_{\sigma\bm{\alpha}}^{-})f^{\prime}(\bm{X})-f(\bm{X})f_{*}(\bm{X}_{\sigma\bm{\alpha}}^{-})]V_{\alpha}\theta(V_{\alpha})d\Omega(\bm{\alpha})d\bm{\xi}d\bm{\xi}_{*}d\bm{X}.\label{eq:I}
\end{equation}
Equation~\eqref{eq:I} can be transformed as
\begin{align}
I(t)= & -\frac{\sigma^{2}}{2m}\int g(\bm{X},\bm{X}_{\sigma\bm{\alpha}}^{-})f_{*}^{\prime}(\bm{X}_{\sigma\bm{\alpha}}^{-})f^{\prime}(\bm{X})V_{\alpha}^{\prime}\theta(-V_{\alpha}^{\prime})d\Omega(\bm{\alpha})d\bm{\xi}d\bm{\xi}_{*}d\bm{X}\nonumber \\
 & -\frac{\sigma^{2}}{2m}\int g(\bm{X},\bm{X}_{\sigma\bm{\alpha}}^{-})f(\bm{X})f_{*}(\bm{X}_{\sigma\bm{\alpha}}^{-})V_{\alpha}\theta(V_{\alpha})d\Omega(\bm{\alpha})d\bm{\xi}d\bm{\xi}_{*}d\bm{X}\displaybreak[0]\nonumber \\
= & \frac{\sigma^{2}}{2m}\int g(\bm{X},\bm{X}_{\sigma\bm{\alpha}}^{-})f(\bm{X})f_{*}(\bm{X}_{\sigma\bm{\alpha}}^{-})[(\bm{\xi}-\bm{\xi}_{*})\cdot\bm{\alpha}]d\Omega(\bm{\alpha})d\bm{\xi}d\bm{\xi}_{*}d\bm{X}\displaybreak[0]\nonumber \\
= & \frac{\sigma^{2}}{2m}\int g(\bm{X},\bm{X}_{\sigma\bm{\alpha}}^{-})\rho(\bm{X})\rho(\bm{X}_{\sigma\bm{\alpha}}^{-})\bm{v}(\bm{X})\cdot\bm{\alpha}d\Omega(\bm{\alpha})d\bm{X}\nonumber \\
 & -\frac{\sigma^{2}}{2m}\int g(\bm{X},\bm{X}_{\sigma\bm{\alpha}}^{-})\rho(\bm{X})\rho(\bm{X}_{\sigma\bm{\alpha}}^{-})\bm{v}(\bm{X}_{\sigma\bm{\alpha}}^{-})\cdot\bm{\alpha}d\Omega(\bm{\alpha})d\bm{X}\displaybreak[0]\nonumber \\
= & \frac{\sigma^{2}}{2m}\int g(\bm{X}_{\sigma\bm{\alpha}}^{+},\bm{X})\rho(\bm{X}_{\sigma\bm{\alpha}}^{+})\rho(\bm{X})\bm{v}(\bm{X}_{\sigma\bm{\alpha}}^{+})\cdot\bm{\alpha}d\Omega(\bm{\alpha})d\bm{X}\nonumber \\
 & +\frac{\sigma^{2}}{2m}\int g(\bm{X},\bm{X}_{\sigma\bm{\alpha}}^{+})\rho(\bm{X})\rho(\bm{X}_{\sigma\bm{\alpha}}^{+})\bm{v}(\bm{X}_{\sigma\bm{\alpha}}^{+})\cdot\bm{\alpha}d\Omega(\bm{\alpha})d\bm{X}\displaybreak[0]\nonumber \\
= & \frac{\sigma^{2}}{m}\int g(\bm{X},\bm{X}_{\sigma\bm{\alpha}}^{+})\rho(\bm{X})\rho(\bm{X}_{\sigma\bm{\alpha}}^{+})\bm{v}(\bm{X}_{\sigma\bm{\alpha}}^{+})\cdot\bm{\alpha}d\Omega(\bm{\alpha})d\bm{X},\label{eq:5.4}
\end{align}
where $V_{\alpha}^{\prime}\equiv(\bm{\xi}_{*}^{\prime}-\bm{\xi}^{\prime})\cdot\bm{\alpha}=-V_{\alpha}$
is used at the first equality, (III) is used at the second equality,
the integration with respect to $\bm{\xi}$ and $\bm{\xi}_{*}$ is
performed at the third equality, and the shift operation by $+\sigma\bm{\alpha}$
and (II) are used at the fourth equality. The last line of (\ref{eq:5.4})
is further transformed as
\begin{align}
I(t)= & \frac{\sigma^{2}}{m}\int g(\bm{X},\bm{X}_{\sigma\bm{\alpha}}^{+})\rho(\bm{X})\rho(\bm{X}_{\sigma\bm{\alpha}}^{+})\bm{v}(\bm{X}_{\sigma\bm{\alpha}}^{+})\cdot\bm{\alpha}d\Omega(\bm{\alpha})d\bm{X}\displaybreak[0]\notag\\
= & \frac{\sigma^{2}}{m}\int\delta(|\bm{X}-\bm{Y}|-\sigma)g(\bm{X},\bm{Y})\rho(\bm{X})\rho(\bm{Y})\bm{v}(\bm{Y})\cdot\frac{\bm{Y}-\bm{X}}{\sigma^{2}|\bm{Y}-\bm{X}|}d\bm{Y}d\bm{X}\displaybreak[0]\notag\\
= & \frac{1}{m}\int g(\bm{X},\bm{Y})\rho(\bm{X})\rho(\bm{Y})\bm{v}(\bm{Y})\cdot\frac{\partial}{\partial\bm{Y}}\theta(|\bm{X}-\bm{Y}|-\sigma)d\bm{Y}d\bm{X}\displaybreak[0]\notag\\
= & \frac{1}{m}\int_{D\times D}g_{2}(\bm{X},\bm{Y})\rho(\bm{X})\rho(\bm{Y})\bm{v}(\bm{X})\cdot\frac{\partial}{\partial\bm{X}}\theta(|\bm{X}-\bm{Y}|-\sigma)d\bm{X}d\bm{Y},
\end{align}
and the last line is reduced by (\ref{eq:reduction}) in Appendix~\ref{sec:Correlation-function} to
\begin{align}
I(t) & =\int_{D}\rho\bm{v}\cdot\frac{\partial}{\partial\bm{X}}{\ln\frac{\rho}{w}}d\bm{X}\notag\displaybreak[0]\\
 & =-\int_{\partial D}\rho\bm{v}\cdot\bm{n}{\ln\frac{\rho}{w}}dS-\int_{D}({\ln\frac{\rho}{w}})\frac{\partial}{\partial\bm{X}}\cdot(\rho\bm{v})d\bm{X}\notag\displaybreak[0]\\
 & =\int_{D}\frac{\partial\rho}{\partial t}{\ln\frac{\rho}{w}}d\bm{X}\notag\displaybreak[0]\\
 & =\frac{d}{dt}\int_{D}\rho({\ln\frac{\rho}{w}}-1)d\bm{X}+\int_{D}\frac{\rho}{w}\frac{\partial w}{\partial t}d\bm{X}\notag\displaybreak[0]\\
 & =\frac{d}{dt}\Big(\int_{D}\rho{\ln\frac{\rho}{w}}d\bm{X}+m\ln\phi\Big).
\end{align}
Here $\bm{v}\cdot\bm{n}=0$ on $\partial D$, the continuity equation
(\ref{eq:continuity}), and the relation
\begin{align}
\frac{1}{\phi}\frac{d\phi}{dt}= & \frac{N}{\phi}\int_{D^{N}}\frac{\partial w(\bm{X}_{1})}{\partial t}w(\bm{X}_{2})\cdots w(\bm{X}_{N})\Theta(\bm{X}_{1},\cdots,\bm{X}_{N})d\bm{X}_{1}\cdots d\bm{X}_{N}\displaybreak[0]\nonumber \\
= & \frac{1}{m}\int_{D}\frac{\partial w(\bm{X}_{1})}{\partial t}\frac{\rho(\bm{X}_{1})}{w(\bm{X}_{1})}d\bm{X}_{1},\label{eq:phi_vs_w}
\end{align}
have been used; {see \eqref{eq:A.1}, \eqref{eq:A.2}, and \eqref{eq:Y}} in Appendix~\ref{sec:Correlation-function}, as for (\ref{eq:phi_vs_w}). Hence,
we finally arrive at
\begin{equation}
I(t)=-\frac{d\mathcal{H}^{(c)}}{dt},\label{I(t)}
\end{equation}
where $\mathcal{H}^{(c)}$ is a so-called collisional part of the H function
defined by
\begin{equation}
 \mathcal{H}^{(c)}(t)
 \equiv
-\int_{D}\rho(\bm{X}){\ln\frac{\rho(\bm{X})}{w(\bm{X})}}d\bm{X}-m\ln\phi.\label{eq:Hc}
\end{equation}
The total H function $\mathcal{H}\equiv\mathcal{H}^{(k)}+\mathcal{H}^{(c)}$
thus satisfies the following inequality:
\begin{equation}
\frac{d\mathcal{H}}{dt}+\int_{D}\frac{\partial}{\partial X_{i}}\langle\xi_{i}f{\ln f}\rangle d\bm{X}\le0,\label{eq:Htheorem}
\end{equation}
where the equality holds if and only if $f_{*}^{\prime}(\bm{X}_{\sigma\bm{\alpha}}^{-})f^{\prime}(\bm{X})=f_{*}(\bm{X}_{\sigma\bm{\alpha}}^{-})f(\bm{X})$.

{
\begin{remark}\label{rem2a}
The above $\mathcal{H}$ is bounded. See Appendix~\ref{boundedness}.
\end{remark}}

\begin{remark}\label{rem2}
If the system is isolated, the second term on the left-hand side of
(\ref{eq:Htheorem}) vanishes, and $\mathcal{H}$ monotonically decreases
in time as shown in \cite{MGB18}. Therefore, $-R\mathcal{H}$ is
identified as a natural extension of the thermodynamic entropy to
the case of non-equilibrium state. Equation~(\ref{eq:Htheorem})
combined with the following lemma, i.e., Lemma~\ref{lem:(Darrozes=002013Guiraud)-If-the}, can
be found in \cite[p.~270]{DBK21}.
\end{remark}

\begin{lemma}
\label{lem:(Darrozes=002013Guiraud)-If-the}(Darrozes--Guiraud \cite{DG66,C88,S07})
If the velocity distribution function $f$ satisfies the boundary
condition \eqref{KBC}, then it holds that
\begin{equation}
\int_{\partial D}\langle(\bm{\xi}\cdot\bm{n})f\ln\frac{f}{f_{w}}\rangle dS\le0,
\end{equation}
where $\bm{n}$ is the inward unit normal to the surface $\partial D$
and the equality holds if and only if $f=f_{w}$.
\end{lemma}

\section{Main results: Free energy and its monotonicity\label{sec:Main-results:-Helmholtz}}

After the presentation of the known results \cite{R78,MGB18,BB18} in Sec.~\ref{sec:H-function},
we now discuss the thermal relaxation of a dense gas in a closed system
with the aid of Lemma~\ref{lem:E transfer}.
Consider the multiplication of $1+\ln(f/f_{w})$ with the modified
Enskog equation (\ref{eq:2.1}) and integrate it with respect to $\bm{\xi}$.
Since $f_{w}$ depends on neither $t$ nor $\bm{X}$, we have
\begin{equation}
\frac{\partial}{\partial t}\langle f\ln(f/f_{w})\rangle+\frac{\partial}{\partial X_{i}}\langle\xi_{i}f\ln(f/f_{w})\rangle=\langle\ln(f/f_{w})J_{ME}(f)\rangle.\label{eq:5.8}
\end{equation}
Since ${\ln f_{w}}=a_{w}-\bm{\xi}^{2}/(2RT_{w})$
with $a_{w}$ being a constant, the right-hand side of (\ref{eq:5.8})
is reduced to 
\begin{equation}
\langle\ln(f/f_{w})J_{ME}(f)\rangle=\langle J_{ME}(f){\ln f}\rangle+\frac{1}{2RT_{w}}\langle\bm{\xi}^{2}J_{ME}(f)\rangle.
\end{equation}
Once we integrate (\ref{eq:5.8}) with respect to $\bm{X}$ over the
domain $D$, the contribution from $\langle\bm{\xi}^{2}J_{ME}(f)\rangle$
vanishes by Lemma~\ref{lem:E transfer} and we arrive at
\begin{align}
\frac{d}{dt}\int_{D}\langle f\ln(f/f_{w})\rangle d\bm{X}= 
& \int_{\partial D}\langle\xi_{i}n_{i}f\ln(f/f_{w})\rangle dS
 +\int_{D}\langle J_{ME}(f){\ln f}\rangle d\bm{X}\nonumber \\
\le & \int_{\partial D}\langle\xi_{i}n_{i}f\ln(f/f_{w})\rangle dS-\frac{d\mathcal{H}^{(c)}}{dt}\le-\frac{d\mathcal{H}^{(c)}}{dt},
\end{align}
where Lemma~\ref{lem:(Darrozes=002013Guiraud)-If-the} has been used
at the last inequality. By transposing the most right-hand side to
the left-hand side, it is seen that $\mathcal{F}$ defined by 
\begin{equation}
\mathcal{F}\equiv RT_{w}(\int_{D}\langle f\ln(f/f_{w})\rangle d\bm{X}+\mathcal{H}^{(c)}),
\label{Free}
\end{equation}
decreases monotonically in time:
\begin{equation}
\frac{d\mathcal{F}}{dt}\le0,
\end{equation}
where the equality holds if and only if $f_{*}^{\prime}(\bm{X}_{\sigma\bm{\alpha}}^{-})f^{\prime}(\bm{X})=f_{*}(\bm{X}_{\sigma\bm{\alpha}}^{-})f(\bm{X})$
for $\bm{X},\ \bm{X}_{\sigma\bm{\alpha}}^{-}\in D$ and $f=f_{w}$
on $\partial D$; see the equality condition for (\ref{eq:Htheorem})
and in Lemma~\ref{lem:(Darrozes=002013Guiraud)-If-the}.
{Since $\mathcal{F}$ is bounded from below (see Appendix~\ref{boundedness}),
 $\mathcal{F}$ approaches a stationary value as $t\to\infty$.
The extension to the case of the modified Enskog--Vlasov equation
is discussed in Appendix~\ref{EV}.}

\begin{theorem}
(thermal relaxation in a closed system surrounded by a heat bath)
Suppose that the behavior of a dense gas in a closed system surrounded
by a heat bath with a constant temperature $T_{w}$ is described by
the modified Enskog equation \eqref{MEE} and the boundary condition
\eqref{KBC}. Then a quantity $\mathcal{F}$ defined by 
\begin{equation}
\mathcal{F}=RT_{w}(\int_{D}\langle f\ln(f/f_{w})\rangle d\bm{X}+\mathcal{H}^{(c)}),
\end{equation}
monotonically decreases in time {and approaches a stationary value as $t\to\infty$,} where $f_{w}$ and $\mathcal{H}^{(c)}$
are respectively defined by (\ref{eq:fw}) and (\ref{eq:Hc}).
\end{theorem}

\begin{remark}
From \eqref{Free} and \eqref{H_kinetic}, $\mathcal{F}$ can be rewritten as
\begin{align}
\mathcal{F}=
 & RT_w(\int_D\langle f{\ln f}\rangle d\bm{X}
       -\int_D\langle f{\ln f_w}\rangle d\bm{X}+\mathcal{H}^{(c)})\notag\\
=& (\mathcal{H}^{(k)}+\mathcal{H}^{(c)})RT_w
+\int_D\langle \frac{1}{2}\bm{\xi}^2f\rangle d\bm{X}+\mathrm{const.}\label{eq:F-H}
\end{align}
Since $\int_{D}\langle\frac{1}{2}\bm{\xi}^{2}f\rangle d\bm{X}$ and
$-\mathcal{H}R$ are respectively the internal energy $E$ and the
entropy $S$ of the closed system (see Remark~\ref{rem2}), $\mathcal{F}$ is identified as
$E-T_{w}S$ up to an additive constant, i.e., an extension of the Helmholtz free energy in thermodynamics
to a non-equilibrium system. The present result shows that the same
statement for the Boltzmann equation mentioned in \cite[p.~270]{DBK21}
holds for the modified Enskog equation, thanks to Lemma~\ref{lem:E transfer}.
In the case of the Boltzmann equation, the consideration of Lemma~\ref{lem:E transfer}
was not required.
\end{remark}

{
When $d\mathcal{F}/dt=0$, two conditions
\begin{subequations}\label{eq:conditions}\begin{align}
\ln f_{*}^{\prime}(\bm{X}_{\sigma\bm{\alpha}}^{-})+\ln f^{\prime}(\bm{X})=\ln f_{*}(\bm{X}_{\sigma\bm{\alpha}}^{-})+\ln f(\bm{X}),\quad \textrm{for\ }\bm{X},\bm{X}^-_{\sigma\bm{\alpha}}\in D, \label{eq:cond1}\\
f(t,\bm{X},\bm{\xi})=\frac{\rho(t,\bm{X})}{(2\pi RT_w)^{3/2}}\exp(-\frac{\bm{\xi}^2}{2RT_w}),
\quad \textrm{for\ } \bm{X}\in \partial D,\label{eq:cond2}
\end{align}\end{subequations}
hold. 
On condition
that \eqref{eq:cond1} is identical to
\begin{equation}
\ln f(t,\bm{X},\bm{\xi})=b_{0}(t,\bm{X})+b_{i}(t)\xi_{i}+b_{4}(t)\bm{\xi}^{2}+c_{i}(t)\epsilon_{ijk}X_{j}\xi_{k},
\end{equation}
or equivalently to
\begin{equation}
f(t,\bm{X},\bm{\xi}) =\frac{\rho(t,\bm{X})}{(2\pi RT(t))^{3/2}}\exp(-\frac{(\bm{\xi}-\bm{v}(t,\bm{X}))^{2}}{2RT(t)}),\label{eq:restMax}
\end{equation}
with $\bm{v}(t,\bm{X})=\bm{V}(t)+\bm{X}\times\bm{W}(t)$  \cite{MGB18},
\eqref{eq:cond2} leads to $T(t)=T_{w}$ and $\bm{v}(t,\bm{X})=0$.
Furthermore, $\rho$ is independent of $t$ because of the continuity equation (\ref{eq:continuity}) with $\bm{v}=\bm{0}$. Therefore, when
$d\mathcal{F}/dt=0$,
$f$ is a time-independent resting Maxwellian
\begin{equation}
\frac{\rho(\bm{X})}{(2\pi RT_{w})^{3/2}}\exp(-\frac{\bm{\xi}^{2}}{2RT_{w}}),
\end{equation}
which represents the thermal equilibrium state with the heat bath characterized
by the uniform temperature $T_{w}$.}

\section{Conclusion\label{sec:Conclusion}}

{In the present work,} the thermal relaxation of a dense gas in a closed system surrounded by
a heat bath has been studied on the basis of the modified Enskog equation.
The H theorem established by Resibois \cite{R78} for the infinite
domain and for a periodic domain and then later by Maynar et al. \cite{MGB18}
for a bounded domain surrounded by the specular-reflection wall has been
arranged in a form suitable for a closed system surrounded by a heat bath.
The case of the modified Enskog--Vlasov equation
{has also been} considered in Appendix~\ref{EV}.
{Different} from the case of the Boltzmann equation, it is required to pay attention to
collisional contributions to the momentum and the energy transport.
We have confirmed, however, that their net contributions on the boundary vanish.
It is physically natural in view of the origin of those transports. As
the result, the Darrozes--Guiraud inequality {plays the same role
as} in the case of the Boltzmann equation to find a quantity $\mathcal{F}$
that {corresponds to} the Helmholtz free energy in the thermodynamics.
This quantity has been shown to {be bounded and to} decrease monotonically in time.

\appendix

\section{{\textit{N}-particle distribution and} correlation function $g_{2}$\label{sec:Correlation-function}}

In the case of the modified Enskog equation, {the $N$-particle (factorized) distribution function $\rho_N$ is introduced:
\begin{equation}
\rho_N=\frac{1}{\phi(t)}\Theta(\bm{X}_1,\cdots,\bm{X}_N)
       W(t,\bm{X}_1,\bm{\xi}_1)\cdots W(t,\bm{X}_N,\bm{\xi}_N),
\end{equation}
and} the velocity distribution function $f$ is expressed in terms of $\rho_N$:
\begin{align}
  f(t,\bm{X}_{1},\bm{\xi}_{1})
=&mN\int_{(D\times\mathbb{R}^3)^{(N-1)}}\rho_N(t,\bm{Z}_{1},\dots,\bm{Z}_{N})d\bm{Z}_{2}\cdots d\bm{Z}_{N}\notag\\
=&{ \frac{mN}{\phi(t)}W(t,\bm{X}_{1},\bm{\xi}_{1}) Y(t,\bm{X}_1),}\label{eq:factorize}
\end{align}{
where and in what follows $\bm{Z}_i=(\bm{X}_i,\bm{\xi}_i)$, $(D\times\mathbb{R}^3)^{N}$ (or $D^{N}$) is the $N$-times direct multiple of $D\times\mathbb{R}^3$ (or $D$),} $N$ is
the number of molecules in $D$,
and
\begin{subequations}\label{eq:partition}\begin{align}
& {Y(t,\bm{X}_1)=\int_{D^{N-1}}w(t,\bm{X}_{2})\cdots w(t,\bm{X}_{N})\Theta(\bm{X}_{1},\cdots,\bm{X}_{N})d\bm{X}_{2}\cdots d\bm{X}_{N},}\label{eq:A.1}\\
 & \phi(t)=\int_{D^{N}}w(t,\bm{X}_{1})\cdots w(t,\bm{X}_{N})\Theta(\bm{X}_{1},\cdots,\bm{X}_{N})d\bm{X}_{1}\cdots d\bm{X}_{N},\displaybreak[0]\label{eq:A.2}\\
 & w(t,\bm{X})=\int W(t,\bm{X},\bm{\xi})d\bm{\xi},\displaybreak[0]\label{eq:A.3}\\
 & \Theta(\bm{X}_{1},\cdots,\bm{X}_{N})=\prod_{i=1}^{N}\prod_{j>i}^{N}\theta(|\bm{X}_{ij}|-\sigma),\quad\bm{X}_{ij}=\bm{X}_{i}-\bm{X}_{j}.\label{eq:A.4}
\end{align}\end{subequations}
{Note that $\rho_N$ is normalized as
\begin{equation}
\int_{(D\times\mathbb{R}^{3})^N} \rho_N d\bm{Z}_1\cdots d\bm{Z}_N=1,
\end{equation}
and the density $\rho$ is also expressed as
\begin{equation}
\rho(t,\bm{X})=\frac{mN}{\phi(t)}w(t,\bm{X})Y(t,\bm{X}),\label{eq:Y}
\end{equation}
by a simple integration of \eqref{eq:factorize} with respect to $\bm{\xi}_1$.}

The correlation function $g_{2}$ {in \eqref{eq:g_g2}} is then defined in terms of the quantities in \eqref{eq:partition} as\footnote{In the literature, $\Theta$ is often used in place of $\Theta_{(1,2)}$ in the definition of $g_2$. The definition (\ref{eq:g2}) is adopted in order to avoid any ambiguity occurring
in the derivation of \eqref{I(t)}.
}\begin{subequations}\label{eq:PairFunc}
\begin{align}
 & g_{2}(t,\bm{X}_{1},\bm{X}_{2})\nonumber \\
= & \frac{m^{2}N(N-1)}{\phi(t)}\frac{w(t,\bm{X}_{1})w(t,\bm{X}_{2})}{\rho(t,\bm{X}_{1})\rho(t,\bm{X}_{2})}\nonumber \\
 & \quad\times\int_{D^{N-2}}w(t,\bm{X}_{3})\cdots w(t,\bm{X}_{N})\Theta_{(1,2)}(\bm{X}_{1},\cdots,\bm{X}_{N})d\bm{X}_{3}\cdots d\bm{X}_{N},\label{eq:g2}
\end{align}
where
\begin{equation}
\Theta_{(1,2)}(\bm{X}_{1},\cdots,\bm{X}_{N})=\prod_{i=1}^{N}\prod_{j>\max(i,2)}^{N}\theta(|\bm{X}_{ij}|-\sigma).\label{eq:Theta12}
\end{equation}
Note that
\begin{equation}
\Theta(\bm{X}_{1},\cdots,\bm{X}_{N})=\theta(|\bm{X}_{12}|-\sigma)\Theta_{(1,2)}(\bm{X}_{1},\cdots,\bm{X}_{N}),\label{eq:A.7}
\end{equation}
\end{subequations}by \eqref{eq:A.4} {and \eqref{eq:Theta12}.
By \eqref{eq:Y} with \eqref{eq:A.1},
$\rho$ can be regarded as a functional of $w$ and, if invertible, vice versa. Hence, $\phi$ and $g_{2}$ can also be regarded as functionals of $\rho$. 
It is seen from (\ref{eq:A.7}) that} 
\begin{align}
 & \int_{D^{N-1}}\Theta_{(1,2)}(\bm{X}_{1},\dots,\bm{X}_{N})\frac{\partial}{\partial\bm{X}_{1}}\theta(|\bm{X}_{12}|-\sigma)F(\bm{X}_{2},\dots,\bm{X}_{N})d\bm{X}_{2}\dots d\bm{X}_{N}\nonumber \\
= & \frac{1}{N-1}\frac{\partial}{\partial\bm{X}_{1}}\int_{D^{N-1}}\Theta_{(1,2)}(\bm{X}_{1},\dots,\bm{X}_{N})\theta(|\bm{X}_{12}|-\sigma)F(\bm{X}_{2},\dots,\bm{X}_{N})d\bm{X}_{2}\dots d\bm{X}_{N}\nonumber \\
= & \frac{1}{N-1}\frac{\partial}{\partial\bm{X}_{1}}\int_{D^{N-1}}\Theta(\bm{X}_{1},\dots,\bm{X}_{N})F(\bm{X}_{2},\dots,\bm{X}_{N})d\bm{X}_{2}\dots d\bm{X}_{N},\label{eq:formula}
\end{align}
if $F(\bm{X}_{2},\dots,\bm{X}_{N})$ is a function such that 
\begin{equation}
F(\bm{X}_{2},\dots,\bm{X}_{i}\dots,\bm{X}_{j}\dots,\bm{X}_{N})=F(\bm{X}_{2},\dots,\bm{X}_{j}\dots,\bm{X}_{i}\dots,\bm{X}_{N}),
\end{equation}
for $\forall i,j\in\{2,\dots,N\}$.

{Now, thanks to (\ref{eq:formula}),
the reduction used in Sec.~\ref{sec:H-function} is possible as follows:}
\begin{align}
 & \frac{1}{m}\int_{D}\rho(\bm{X}_{2})\rho(\bm{X}_{1})g_{2}(\bm{X}_{1},\bm{X}_{2})\frac{\partial}{\partial\bm{X}_{1}}\theta(|\bm{X}_{12}|-\sigma)d\bm{X}_{2}\nonumber \\
= & \frac{mN(N-1)}{\phi(t)}w(\bm{X}_{1})\int_{D^{N-1}}w(\bm{X}_{2})\cdots w(\bm{X}_{N})\nonumber \\
 & \quad\times\Theta_{(1,2)}(\bm{X}_{1},\cdots,\bm{X}_{N})\frac{\partial}{\partial\bm{X}_{1}}\theta(|\bm{X}_{12}|-\sigma)d\bm{X}_{2}\cdots d\bm{X}_{N}\displaybreak[0]\nonumber \\
= & w(\bm{X}_{1})\frac{\partial}{\partial\bm{X}_{1}}\{\frac{mN}{\phi(t)}\int_{D^{N-1}}w(\bm{X}_{2})\cdots w(\bm{X}_{N})\Theta(\bm{X}_{1},\cdots,\bm{X}_{N})d\bm{X}_{2}\cdots d\bm{X}_{N}\}\displaybreak[0]\nonumber \\
= & w(\bm{X}_{1})\frac{\partial}{\partial\bm{X}_{1}}\frac{\rho(\bm{X}_{1})}{w(\bm{X}_{1})}=\rho(\bm{X}_{1})\frac{\partial}{\partial\bm{X}_{1}}\ln\frac{\rho(\bm{X}_{1})}{w(\bm{X}_{1})},\label{eq:reduction}
\end{align}
where {\eqref{eq:g2}, \eqref{eq:A.1}, and \eqref{eq:Y} have been used} and the argument $t$ is
omitted from $\rho$ and $w$.

{
\section{Boundedness of $\mathcal{F}$\label{boundedness}}
In this Appendix, we will show that $\mathcal{F}$ is bounded.

With the preparations in Appendix~\ref{sec:Correlation-function}, we first show that $\mathcal{H}$ occurring in \eqref{eq:Htheorem} is identical to the following $H$: \cite{R78,MGB18}
\begin{equation}
H(t)=m\int_{(D\times\mathbb{R}^{3})^N} \rho_N\ln\rho_N d\bm{Z}_1\cdots d\bm{Z}_N.\label{eq:Horiginal}
\end{equation}
Indeed, since $\Theta\ln\Theta\equiv 0$,
the integrations with respect to $\bm{Z}_2,\cdots,\bm{Z}_N$ are simplified to yield   
\begin{align}
H(t)=&m\int_{D^N\times\mathbb{R}^{3N}} \rho_N(\sum_{i=1}^N\ln W(t,\bm{X}_i,\bm{\xi}_i)-\ln\phi) d\bm{X}_1\cdots d\bm{X}_Nd\bm{\xi}_1\cdots d\bm{\xi}_N \notag\displaybreak[0]\\
=&\int_{D\times\mathbb{R}^{3}} f(t,\bm{X}_1,\bm{\xi}_1)\ln W(t,\bm{X}_1,\bm{\xi}_1) d\bm{X}_1d\bm{\xi}_1 -m\ln\phi.\label{eq:HH}
\end{align}
Because of \eqref{eq:factorize} and \eqref{eq:Y}, 
\begin{equation}
\ln W=\ln f-\ln \frac{\rho}{w},\label{eq:lnW}
\end{equation}
and substitution to \eqref{eq:HH} leads to
\begin{align}
H(t)=&\mathcal{H}^{(k)}-\int_{D\times\mathbb{R}^{3}} f(t,\bm{X}_1,\bm{\xi}_1)\ln\frac{\rho(t,\bm{X}_1)}{w(t,\bm{X}_1)} d\bm{X}_1d\bm{\xi}_1 -m\ln\phi\notag\\
 =&\mathcal{H}^{(k)}-\int_{D} \rho(t,\bm{X}_1)\ln\frac{\rho(t,\bm{X}_1)}{w(t,\bm{X}_1)} d\bm{X}_1 -m\ln\phi \notag\\
 =&\mathcal{H}^{(k)}+\mathcal{H}^{(c)}=\mathcal{H}.
\end{align}

Now, thanks to the form  \eqref{eq:Horiginal}, 
the same method as the case of the Boltzmann equation (see, e.g., \cite[Sec.~9.4]{CIP94})
is available to show that $\mathcal{F}$ is bounded from below, which is as follows.
As $x$ increases from $x=0$, $x\ln x$ first monotonically decreases and reaches the minimum at $x=e^{-1}$, and then increases monotonically for $x>e^{-1}$.
Hence, if $\rho_N\ge e^{-1}$, $\rho_N\ln\rho_N \ge -\rho_N$. 
If $\rho_N< e^{-1}$, we split this case into
(i) $\rho_N\ge(4\pi RT_w)^{-3N/2}V_D^{-N}\exp(-\sum_{i=1}^N \frac{\bm{\xi}_i^2}{4RT_w})$ and 
(ii) $\rho_N<(4\pi RT_w)^{-3N/2}V_D^{-N}\exp(-\sum_{i=1}^N \frac{\bm{\xi}_i^2}{4RT_w})$,
where $V_D$ is the volume of $D$.
In case (i), $\rho_N\ln\rho_N \ge \rho_N[-(3N/2)\ln(4\pi RT_w)-N\ln V_D-\sum_{i=1}^N \frac{\bm{\xi}_i^2}{4RT_w}]$; in case (ii), $\rho_N\ln\rho_N>(4\pi RT_w)^{-3N/2}V_D^{-N}\exp(-\sum_{i=1}^N \frac{\bm{\xi}_i^2}{4RT_w}) [-(3N/2)\ln(4\pi RT_w)-N\ln V_D-\sum_{j=1}^N \frac{\bm{\xi}_j^2}{4RT_w}]$. Consequently, it holds that
\begin{align}
\rho_N\ln\rho_N\ge
&-\rho_N-\rho_N N\ln[(4\pi RT_w)^{3/2}V_D]-\rho_N\sum_{j=1}^N \frac{\bm{\xi}_j^2}{4RT_w}\notag
\displaybreak[0]\\
&-\sum_{i=1}^N \frac{\bm{\xi}_i^2}{4RT_w}
  \frac{1}{(4\pi RT_w)^{3N/2}V_D^N}\exp(-\sum_{j=1}^N \frac{\bm{\xi}_j^2}{4RT_w}) \displaybreak[0]\notag\\
&-\frac{N\ln[(4\pi RT_w)^{3/2}V_D]}{(4\pi RT_w)^{3N/2}V_D^N}\exp(-\sum_{i=1}^N \frac{\bm{\xi}_i^2}{4RT_w}),
\end{align}
by which $H$ is evaluated as 
\begin{align}
H(t)=&m\int_{(D\times \mathbb{R}^3)^N} \rho_N\ln\rho_N d\bm{Z}_1\cdots d\bm{Z}_N \notag \\
  \ge&-m\int_{(D\times \mathbb{R}^3)^N}
  \{\rho_N+\rho_NN\ln[(4\pi RT_w)^{3/2}V_D]
   +\rho_N\sum_{j=1}^N \frac{\bm{\xi}_j^2}{4RT_w}\notag \displaybreak[0]\\
&+\sum_{i=1}^N \frac{\bm{\xi}_i^2}{4RT_w}
  \frac{1}{(4\pi RT_w)^{3N/2}V_D^N}\exp(-\sum_{j=1}^N \frac{\bm{\xi}_j^2}{4RT_w}) \displaybreak[0]\notag\\
&+\frac{N\ln[(4\pi RT_w)^{3/2}V_D]}{(4\pi RT_w)^{3N/2}V_D^N}\exp(-\sum_{i=1}^N \frac{\bm{\xi}_i^2}{4RT_w})\}d\bm{Z}_1\cdots d\bm{Z}_N \notag\displaybreak[0]\\
  =&-m-2mN\ln[(4\pi RT_w)^{3/2}V_D]
  -m\{\int_{(D\times \mathbb{R}^3)^N}\rho_N\sum_{j=1}^N \frac{\bm{\xi}_j^2}{4RT_w}\notag
\displaybreak[0]\\
&+\sum_{i=1}^N \frac{\bm{\xi}_i^2}{4RT_w}
  \frac{1}{(4\pi RT_w)^{3N/2}V_D^N}\exp(-\sum_{j=1}^N \frac{\bm{\xi}_j^2}{4RT_w}) \}d\bm{Z}_1\cdots d\bm{Z}_N \notag\displaybreak[0]\\
 =& -\{m+2mN\ln[(4\pi RT_w)^{3/2}V_D]+\int_{D\times \mathbb{R}^3}
  f(t,\bm{X},\bm{\xi}) \frac{\bm{\xi}^2}{4RT_w}d\bm{X}d\bm{\xi}\notag\\
 &+mN\int_{\mathbb{R}^3} \frac{\bm{\xi}^2}{4RT_w}
 \frac{1}{(4\pi RT_w)^{3/2}}\exp(-\frac{\bm{\xi}^2}{4RT_w})d\bm{\xi}
 \}\notag\displaybreak[0]\\
 \ge& -\{mN(\frac52+\ln[(4\pi RT_w)^3V_D^2])+\int_{D\times \mathbb{R}^3}
  f(t,\bm{X},\bm{\xi}) \frac{\bm{\xi}^2}{4RT_w}d\bm{X}d\bm{\xi}\}\notag\displaybreak[0]\\
=& -\frac{1}{2RT_w}\int_D \langle \frac12\bm{\xi}^2f\rangle d\bm{X}
  +\mathrm{const.}\label{eq:estimate}
\end{align}
Remind that $mN$ is the total mass in $D$ and thus is finite.
Hence \eqref{eq:estimate} means that $\mathcal{F}\ge\frac{1}{2}\int_D \langle \frac12\bm{\xi}^2f\rangle d\bm{X}+\mathrm{const.}$ by \eqref{eq:F-H}. Moreover, if $\mathcal{F}$ is initially finite, 
then $\mathcal{F}$, $\mathcal{H}$, and $\int_D \langle \frac12\bm{\xi}^2f\rangle d\bm{X}$ 
are bounded individually from both below and above for $t\ge0$.
}

\section{The case of modified Enskog--Vlasov equation\label{EV}}

In the case of Enskog--Vlasov equation, an external force term $F_{i}{\partial f}/{\partial\xi_{i}}$ is added on the left-hand side of \eqref{MEE}, where
\begin{equation}
F_{i}=-\int_{D}\frac{\partial}{\partial X_{i}}\Phi(|\bm{Y}-\bm{X}|)\rho(t,\bm{Y})d\bm{Y},\label{Vlasov}
\end{equation} 
and $\Phi$ is the attractive isotropic force potential between molecules. 

By taking the $(1+{\ln{f}})$-moment of the external force term:
\begin{equation}
\langle(1+{\ln f})F_{i}\frac{\partial f}{\partial\xi_{i}}\rangle=\langle F_{i}\frac{\partial}{\partial\xi_{i}}(f{\ln f})\rangle=0,
\end{equation}
and thus the external term is found to give no contribution to (\ref{eq:27}).
Hence, (\ref{eq:Htheorem}) remains unchanged.

Next consider the $(1+\ln({f}/f_{w}))$-moment:
\begin{equation}
\langle(1+\ln\frac{f}{f_{w}})F_{i}\frac{\partial f}{\partial\xi_{i}}\rangle=-\langle({\ln{f_{w}}})F_{i}\frac{\partial f}{\partial\xi_{i}}\rangle=F_{i}\langle\frac{\bm{\xi}^{2}}{2RT_{w}}\frac{\partial f}{\partial\xi_{i}}\rangle=-\frac{\rho v_{i}F_{i}}{RT_{w}}.
\end{equation}
Since $F_i$ is given by \eqref{Vlasov}, 
\begin{align}
-\int_{D}\frac{\rho v_{i}F_{i}}{RT_{w}}d\bm{X} 
=& \int_{D}\frac{\rho v_{i}}{RT_{w}}\frac{\partial}{\partial X_{i}}\int_{D}\Phi(|\bm{Y}-\bm{X}|)\rho(t,\bm{Y})d\bm{Y}d\bm{X}\displaybreak[0]\notag\\
=&-\int_{\partial D}\frac{\rho v_{i}}{RT_{w}}n_{i}\int_{D}\Phi(|\bm{Y}-\bm{X}|)\rho(t,\bm{Y})d\bm{Y}dS(\bm{X})\notag\\
&-\int_{D}\frac{1}{RT_{w}}\frac{\partial(\rho v_{i})}{\partial X_{i}}\int_{D}\Phi(|\bm{Y}-\bm{X}|)\rho(t,\bm{Y})d\bm{Y}d\bm{X}\displaybreak[0]\notag\\
=& \int_{D}\frac{1}{RT_{w}}\frac{\partial\rho(t,\bm{X})}{\partial t}\int_{D}\Phi(|\bm{Y}-\bm{X}|)\rho(t,\bm{Y})d\bm{Y}d\bm{X}\displaybreak[0]\notag\\
=& \frac{1}{2}\frac{d}{dt}\int_{D\times D}\frac{\Phi(|\bm{Y}-\bm{X}|)}{RT_{w}}\rho(t,\bm{X})\rho(t,\bm{Y})d\bm{X}d\bm{Y},
\end{align}
where $v_in_i=0$ on $\partial D$ and the continuity equation \eqref{eq:continuity} have been used.
Therefore, in the case of the modified Enskog--Vlasov equation,
\begin{align}
\mathcal{F}^{\prime}\equiv & RT_{w}(\int_{D}\langle f\ln(f/f_{w})\rangle d\bm{X}+\mathcal{H}^{(c)})\notag\\
&+\frac{1}{2}\int_{D\times D}\Phi(|\bm{Y}-\bm{X}|)\rho(t,\bm{X})\rho(t,\bm{Y})d\bm{X}d\bm{Y},
\end{align}
decreases monotonically in time:
\begin{equation}
\frac{d\mathcal{F}^{\prime}}{dt}\le0.
\end{equation}
This {corresponds to} the result 
in Appendix B of \cite{TMH21} for a simple kinetic model.
{If $\Phi\ge C$ holds for some constant $C$, $\mathcal{F}^\prime$ is bounded from below and approaches a stationary value as $t\to\infty$. }

\section*{Acknowledgements}

The present work has been supported in part by the JSPS KAKENHI Grant
(No.~22K03923) and the Kyoto University Foundation.
The author thanks Masanari Hattori for his helpful comments to the draft of this paper.

\medskip
Received xxxx 20xx; revised xxxx 20xx.
\medskip

\end{document}